\documentclass[10pt]{article}
\usepackage{}
\usepackage{amssymb}
\usepackage{amsfonts}
\usepackage{bbm}
\usepackage{mathrsfs}
\usepackage{mathrsfs}
\usepackage{longtable}
\usepackage{supertabular}
\usepackage{amsfonts,mathrsfs,latexsym,amsmath,amssymb,amsthm}

\newtheorem{theorem}{Theorem}

\newtheorem{lemma}{Lemma}

\newtheorem{example}{Example}

\begin{document}
\title{\bf On cyclic codes over $\mathbb{Z}_q+u\mathbb{Z}_q$}
\author{{\bf  Jian Gao$^{*1}$, ~Fang-Wei Fu$^2$,~Ling Xiao$^1$, Rama Krishna Bandi$^3$}\\
   {\footnotesize \emph{1. School of Science, Shandong University of Technology}}\\
  {\footnotesize  \emph{Zibo, 255091, P. R. China}}\\
 {\footnotesize \emph{2. Chern Institute of Mathematics and LPMC, Nankai University}}\\
  {\footnotesize  \emph{Tianjin, 300071, P. R. China}}\\
  {\footnotesize \emph{3. Department of Mathematics,Indian Institute of Technology Roorkee}}\\
  {\footnotesize  \emph{Roorkee, 247667, India}}\\
 {\footnotesize  \emph{E-mail: jiangao@mail.nankai.edu.cn}}}
\date{}

\maketitle \noindent {\small {\bf Abstract} Let $R=\mathbb{Z}_q+u\mathbb{Z}_q$, where $q=p^s$ and $u^2=0$. In this paper, some structural properties of cyclic codes over the ring $R$ are considered. A necessary and sufficient condition for cyclic codes over the ring $R$ to be free is obtained and a BCH-type bound on the minimum Hamming distance for them is
also given.}
\vskip 1mm



\vskip 3mm \baselineskip 0.2in

\section{Introduction}
The study of linear codes over finite rings has started since the 1960s. Many excellent work on linear codes over finite rings emerged after the significant discovery that some good nonlinear binary codes can be viewed as the binary images of some cyclic codes under the Gray map from $\mathbb{Z}_4$ to $\mathbb{F}_2^2$ (see \cite{Hammons}). Since then, many coding theorists pay their more and more attentions to the codes over finite rings. In these studies, the group rings associated with codes are finite chain rings in general.
\par
Recently, some coding theorists considered linear codes over the finite non-chain ring $\mathbb{F}_p+v\mathbb{F}_p+\cdots+v^{m-1}\mathbb{F}_p$, where $v^m=v$ and $m-1$ is a divisor of $p-1$ (see \cite{Bayram,Gao,Kaya,Zhu1,Zhu2}). In \cite{Zhu2}, Zhu et al. studied the cyclic codes over $\mathbb{F}_2+v\mathbb{F}_2$. In the subsequent paper \cite{Zhu1}, they investigated a class of constacyclic codes over $\mathbb{F}_p+v\mathbb{F}_p$. In \cite{Kaya}, the authors used the theory of cyclic codes over $\mathbb{F}_p+v\mathbb{F}_p$ to discussed some structural properties of quadratic residue (QR) codes over $\mathbb{F}_p+v\mathbb{F}_p$. Some good linear codes and self-dual codes over $\mathbb{F}_p$ are obtained by QR codes and their extended codes respectively. In \cite{Gao}, Gao discussed some results on linear codes and cyclic codes over $\mathbb{F}_p+v\mathbb{F}_p+v^2\mathbb{F}_p$ with $p$ an odd prime. Further, Bayram and Siap gave some basic results on linear codes and constacyclic codes over $\mathbb{F}_p+v\mathbb{F}_p+\cdots+v^{p-1}\mathbb{F}_p$ (see \cite{Bayram}). More recently, Yildiz and Karadeniz \cite{Yildiz} studied the linear codes over the non-principal ring $\mathbb{Z}_4+u\mathbb{Z}_4$, where $u^2=0$. They introduced the MacWilliams identities for the complete, symmetrized and Lee weight enumerators. They also gave three methods to construct formally self-dual codes over $\mathbb{Z}_4+u\mathbb{Z}_4$. Bandi and Bhaintwal studied some structural properties of cyclic codes of odd length over $\mathbb{Z}_4+u\mathbb{Z}_4$, where $u^2=0$ (see \cite{Bandi}). They provided the general form of the generators of a cyclic code over $\mathbb{Z}_4+u\mathbb{Z}_4$, and they also determined a necessary and sufficient condition for cyclic codes of odd length over $\mathbb{Z}_4+u\mathbb{Z}_4$ to be $(\mathbb{Z}_4+u\mathbb{Z}_4)$-free. Let $R=\mathbb{Z}_q+u\mathbb{Z}_q$, where $p$ is a prime, $q=p^s$ and $u^2=0$. It is natural to ask if we can also study some structural properties of cyclic codes over $R$. In this paper, we mainly consider this issue.
\par
The paper is organized as follows. In Section 2, we introduce some basic results on the ring $R$. In Section 3, we study the cyclic codes over $R$. Using the Chinese Remainder Theorem, any cyclic code of odd length over $R$ can be viewed as the direct sum of ideals of some Galois extension rings of $R$, which can deduce the enumerator of cyclic codes. In Section 4, we determine the generator of cyclic codes. We also give a necessary and sufficient condition for cyclic codes over $R$ to be $R$-free.
\section{The ring $\mathbb{Z}_q+u\mathbb{Z}_q$}
Let $R=\mathbb{Z}_q+u\mathbb{Z}_q$, where $p$ is some prime, $q=p^s$ and $u^2=0$. If $q=4$, then $R=\mathbb{Z}_4+u\mathbb{Z}_4$. The ring $R$ is isomorphic to the quotient ring $\mathbb{Z}_q[u]/(u^2)$ and $R=\{ a+bu:~a, b \in \mathbb{Z}_q\}$. Further, The ideals of $R$ are of the following forms:
\par
(i)~$(p^i)$ for $0\leq i \leq s$;
\par
(ii)~$(p^ku)$ for $0\leq k\leq s-1$;
\par
(iii)~$(p^j+\alpha u)$ for $1\leq j \leq s-1$ and $\alpha \in \mathbb{F}_p\backslash \{0\}$;
\par
(iv)~$(p^j, u)$ for $1 \leq j \leq s-1$.\\
Therefore, there are $(s-1)(p-1)+3s$ ideals of $R$. For example, there are $1+3\times 2=7$ ideals of $\mathbb{Z}_4+u\mathbb{Z}_4$.
\par
$R$ is a local ring with the characteristic $q$ and the maximal ideal $(p, u)$. But it is not a chain ring, since neither of the ideals $(p)$ or $(u)$ is included each other. Further, $R$ is not principal since the ideal $(p, u)$ can not be generated by any single element of this ideal.
\par
Define a map
\begin{equation}
\begin{split}
^-:~R& \rightarrow R/(p,u) \\
r=a+bu& \mapsto a~({\rm mod}p)
 \end{split}
 \end{equation}
The map $^-$ is a ring homomorphism and $R/(p,u)$ is denoted by the residue field $\overline{R}$. Since for any $r=a+bu$, $a\in \mathbb{Z}_q$ then $\overline{R}$ is isomorphic to the finite field $\mathbb{F}_p$.
\par
Let $R[x]$ be the polynomial ring over $R$. The map $^-$ can be extended to $R[x]$ to $\overline{R}[x]$ in the usual way. The image of any element $f(x)\in R[x]$ under this map is denoted by $\overline{f}(x)$. Two polynomials $f(x)$ and $g(x)$ are said to be coprime over $R$ if and only if there are two polynomials $a(x)$ and $b(x)$ in $R[x]$ such that
\begin{equation}
a(x)f(x)+b(x)g(x)=1.
\end{equation}
A polynomial $f(x)\in R[x]$ is said to be basic irreducible if $\overline{f}(x)$ is irreducible in $\overline{R}[x]$ and basic primitive if $\overline{f}(x)$ is primitive in $\overline{R}[x]$.
\par
In the following, we consider the factorization of $x^n-1$ over $R$. We assume that ${\rm gcd}(n,q)=1$ throughout this paper.
\begin{lemma}
Let $g(x)$ be a irreducible polynomial over $\mathbb{F}_p$, where $g(x)|(x^{p^r-1}-1)$ for some positive integer $r$. Then there exists a unique basic irreducible polynomial $f(x)\in R[x]$ such that $\overline{f}(x)=g(x)$ and $f(x)|(x^{p^r-1}-1)$ over $R$.
\end{lemma}
\begin{proof}
Let $x^{p^r-1}-1=g(x)m(x)$. Since ${\rm gcd}(n,q)=1$, it follows that $g(x)$ has no multiple roots. Clearly, $x\nmid g(x)$. Then $g(x)$ has the unique Hensel lift $f(x)$ over $\mathbb{Z}_q$ such that $f(x)|(x^{p^r-1}-1)$ (see Theorem 13.10 in \cite{Wan}). Since $\mathbb{Z}_q$ is a subring of $R$, it follows that the factorization of $x^{p^r-1}-1$ is still valid over $R$. It means that $f(x)|(x^{p^r-1}-1)$ over $R$. Further, $g(x)$ is irreducible over $\mathbb{F}_p$ deduces $f(x)$ is basic irreducible over $R$.
\end{proof}
We call the polynomial $f(x)$ in Lemma 1 the Hensel lift of $g(x)$ to $R[x]$.
\par
Since ${\rm gcd}(n,q)=1$, it follows that the polynomial $x^n-1$ can be factored uniquely into pairwise coprime basic irreducible polynomials over $R$, i.e.
\begin{equation}
x^n-1=f_1(x)f_2(x)\cdots f_t(x),
\end{equation}
where, for $l=1,2,\ldots, t$, $f_l(x)$ is a basic irreducible polynomial over $R$.
\par
Let $\mathcal {T}=\{0, 1, \xi, \ldots, \xi^{p^m-2}\}$ be the Teichm\"{u}ller set of Galois ring ${\rm GR}(q, m)$, where ${\rm GR}(q, m)$ is the $m$th Galois extension ring of $\mathbb{Z}_q$ and $\xi$ is a basic primitive element of ${\rm GR}(q, m)$. Then for each $a\in {\rm GR}(q, m)$, it can be written as $a=a_0+a_1p+\cdots+a_{s-1}p^{s-1}$ where $a_0, a_1, \ldots, a_{s-1}\in \mathcal {T}$. This is called the $p$-adic representation of the element of ${\rm GR}(q, m)$ (see the Section 3 of Chapter 14 in \cite{Wan}).
\par
Now we consider the Galois extension of $R$. Let $f(x)$ be a basic irreducible polynomial of degree $m$ over $R$. The the $m$th Galois extension $R[x]/(f(x))$ is denoted by ${\rm GR}(R, m)$. Let $\alpha$ be a root of $f(x)$. Then $1, \alpha, \ldots, \alpha^{m-1}$ form a set of $R$-free basis and
\begin{equation}
{\rm GR}(R, m)=\{r_0+r_1\alpha+\cdots+r_{m-1}\alpha^{m-1}:~r_0, r_1, \ldots, r_{m-1}\in R\}.
\end{equation}
The ring ${\rm GR}(R, m)$ is a local ring with maximal ideal $((p, u)+(f(x)))$. Its residue field is isomorphic to $\mathbb{F}_{p^m}$. Moreover,
\begin{equation}
{\rm GR}(R, m)\cong {\rm GR}(q, m)[u]/(u^2)\cong {\rm GR}(q, m)+u{\rm GR}(q, m).
\end{equation}
Therefore, for any element $r=a+bu\in {\rm GR}(q, m)$,
\begin{equation}
r=\sum_{i=0}^{s-1}a_ip^i+u\sum_{i=0}^{s-1}b_ip^i,
\end{equation}
where $a_i, b_i\in \mathcal {T}$ and $a=\sum_{i=0}^{s-1}a_ip^i$, $b=\sum_{i=0}^{s-1}b_ip^i$.
\begin{lemma}
For any $r=\sum_{i=0}^{s-1}a_ip^i+u\sum_{i=0}^{s-1}b_ip^i\in {\rm GR}(R, m)$, $r$ is a unit under multiplication if and only if $a_0\neq 0$.
\end{lemma}
\begin{proof}
One can verify that $r$ is a unit under multiplication of $R$ if and only if $r^k$ is a unit under multiplication of $R$ for any positive integer $k$. Particularly, this is valid for $k=q$. Note that, for any $r\in R$, $r^q=a_0^q\in \mathcal {T}$. Therefore $r$ is a unit if and only if $a_0^q$ is a unit if and only if $a_0^q\neq 0$ if and only if $a_0\neq 0$.
\end{proof}
From Lemma 2, we have that the group of units of ${\rm GR}(R, m)$ denoted by ${\rm GR}(R, m)^*$ is given by
\begin{equation}
{\rm GR}(R, m)^*=\{ \sum_{s=0}^{s-1}a_ip^i+u\sum_{i=0}^{s-1}b_ip^i:~a_i, b_i \in \mathcal {T}, a_0\neq 0\}.
\end{equation}
Let
\begin{equation*}
G_C=\{1, \xi, \ldots, \xi^{p^m-2}\}
\end{equation*}
and
\begin{equation*}
G_A=\{1+\sum_{j=1}^{s-1}a_jp^j+u\sum_{i=0}^{s-1}b_ip^i:~a_j,b_i\in \mathcal {T}\}.
\end{equation*}
\begin{theorem}
${\rm GR}(R, m)^*=G_C\times G_A$, and $|{\rm GR}(R, m)^*|=(p^m-1)\cdot(p^{2s-1})^m$.
\end{theorem}
\begin{proof}
Let $G_C=\{1, \xi, \ldots, \xi^{p^m-2}\}$. Then $G_C$ is a multiplicative cyclic group of order $p^{m}-1$. Let $r=\sum_{i=0}^{s-1}a_ip^i+u\sum_{i=0}^{s-1}b_ip^i$ be a unit of $R$, i.e. $r=\sum_{i=0}^{s-1}a_ip^i+u\sum_{i=0}^{s-1}b_ip^i \in {\rm GR}(R, m)^*$. Define a group homomorphism as follows
\begin{equation*}
\begin{split}
\Gamma:~{\rm GR}(R, m)^* &\rightarrow G_C \\
r=\sum_{i=0}^{s-1}a_ip^i+u\sum_{i=0}^{s-1}b_ip^i & \mapsto a_0.
 \end{split}
\end{equation*}
Clearly, $\Gamma$ is a surjective map. Then, by the first homomorphism theorem, we have
\begin{equation*}
{\rm GR}(R, m)^*/{\rm Ker}\Gamma \cong G_C.
\end{equation*}
Clearly,
\begin{equation*}
{\rm Ker}\Gamma=\{1+\sum_{j=1}^{s-1}a_jp^j+\sum_{i=0}^{s-1}b_ip^i:~a_1,a_2,\ldots,a_{s-1},b_0,b_1,\ldots,b_{s-1}\in \mathcal {T}\}.
\end{equation*}
Denote ${\rm Ker}\Gamma$ by $G_A$. Then ${\rm GR}(R, m)^*/G_A\cong G_C$ and $|{\rm GR}(R, m)^*|=|G_C|\cdot |G_A|=(p^m-1)\cdot(p^{2s-1})^m$
\end{proof}
The set of all zero divisors of ${\rm GR}(R, m)$ is
\begin{equation}
\{\sum_{j=1}^{s-1}a_jp^j+u\sum_{i=0}^{s-1}b_ip^i\},
\end{equation}
which is the maximal ideal of ${\rm GR}(R, m)$.
\begin{lemma}
Let $f(x)$ and $g(x)$ be polynomials of $R[x]$. Then $f(x)$ and $g(x)$ are coprime if and only if $\overline{f}(x)$ and $\overline{g}(x)$ are coprime over $\overline{R}=\mathbb{F}_p$.
\end{lemma}
\begin{proof}
If $f(x)$ and $g(x)$ are coprime over $R$, then there are polynomials $a(x), b(x)\in R[x]$ such that
\begin{equation*}
a(x)f(x)+b(x)g(x)=1,
\end{equation*}
which implies that
\begin{equation*}
\overline{a}(x)\overline{f}(x)+\overline{b}(x)\overline{g}=1
\end{equation*}
with $\overline{a}(x), \overline{b}(x), \overline{f}(x), \overline{g}(x) \in \mathbb{F}_p[x]$. Therefore $\overline{f}(x)$ and $\overline{g}(x)$ are coprime over $\mathbb{F}_p$.
\par
On the other hand, if $\overline{f}(x)$ and $\overline{g}(x)$ are coprime over $\mathbb{F}_p$, then there are polynomials $\overline{a}(x)$ and  $\overline{b}(x)\in \mathbb{F}_p[x]$ such that
\begin{equation*}
\overline{a}(x)\overline{f}(x)+\overline{b}(x)\overline{g}=1,
\end{equation*}
which implies that
\begin{equation*}
a(x)f(x)+b(x)g(x)=1+pr(x)+ut(x)
\end{equation*}
for some $p(x), t(x)\in R[x]$. Let
\begin{equation*}
\lambda(x)=\sum_{i=0}^{s-1}(-pr(x))^i~{\rm and}~\tau(x)=1-ut(x)\lambda(x).
\end{equation*}
Let $\kappa(x)=\lambda(x)\tau(x)$. Then
\begin{equation*}
\kappa(x)a(x)f(x)+\kappa(x)b(x)g(x)=1,
\end{equation*}
which implies that $f(x)$ and $g(x)$ are coprime over $R$.
\end{proof}
\begin{theorem} \label{theorem 2.2}
Let ${\rm GR}(R, m)=R[x]/(f(x))$ be the $m$th Galois extension of $R$, where $f(x)$ is a basic irreducible polynomial with degree $m$ over $R$. Then the ideals of ${\rm GR}(R, m)$ are precisely
\par
(i)~$(p^i+(f(x)))$ for $0\leq i \leq s$;
\par
(ii)~$(p^ku+(f(x)))$ for $0\leq k\leq s-1$;
\par
(iii)~$(p^j+\alpha u+(f(x)))$ for $1\leq j \leq s-1$, $\alpha \in \mathbb{F}_{p^m}\backslash \{0\}$;
\par
(iv)~$((p^j, u)+(f(x)))$ for $1 \leq j \leq s-1$.
\end{theorem}
\begin{proof}
Let $I$ be an ideal of $R$. If $I$ is zero, then $I=(p^s+(f(x)))$ of $R[x]/(f(x))$. In the following, we determine the nonzero ideals of $R[x]/(f(x))$. Let $g(x)\in I$. Since $f(x)$ is basic irreducible over $R$, it follows that $\overline{f}(x)$ is irreducible over $\mathbb{F}_p$. Therefore ${\rm gcd}(\overline{f}(x), \overline{g}(x))=1$ or $\overline{f}(x)$. If ${\rm gcd}(\overline{f}(x), \overline{g}(x))=1$, then ${\rm gcd}(f(x), g(x))=1$ which implies that there are polynomials $a(x), b(x)\in R[x]$ such that $a(x)f(x)+bg(x)=1$. It means that $g(x)$ is a unit of $R[x]/(f(x))$, i.e. $I=R[x]/(f(x))$. If ${\rm gcd}(\overline{f}(x), \overline{g}(x))=\overline{f}(x)$, then there are polynomials $a(x), b(x), c(x)\in R[x]$ such that $g(x)=a(x)f(x)+pb(x)+uc(x)$. Therefore $g(x)\in ((p, u)+(f(x)))$ of $R[x]/(f(x))$. Since the ideals contained in $g(x)\in ((p, u)+(f(x)))$ are as the form in this theorem, the result follows.
\end{proof}
\section{Cyclic codes over $R$}
Let $R^n$ be a free $R$-module of rank $n$, i.e. $R^n=\{(c_0, c_1, \ldots, c_{n-1}):~c_0, c_1, \ldots, c_{n-1}\in R\}$. Let $\mathscr{C}$ be a nonempty set of $R^n$. $\mathscr{C}$ is called a linear code of length $n$ if and only if $\mathscr{C}$ is an $R$-submodule of $R^n$. Let $T$ be the cyclic shift operator. If for any $c=(c_0, c_1, \ldots, c_{n-1})\in \mathscr{C}$ the $T(c)=(c_{n-1}, c_0, \ldots, c_{n-2})$ is also in $\mathscr{C}$, we say $\mathscr{C}$ is a cyclic code of length $n$ over $R$. Define an $R$-module isomorphism as follows
\begin{equation*}
\begin{split}
\Phi:~R^n&\rightarrow R[x]/(x^n-1)\\
(c_0, c_1, \ldots, c_{n-1})&\mapsto c_0+c_1x+\cdots+c_{n-1}x^{n-1}.
 \end{split}
 \end{equation*}
One can verify that $\mathscr{C}$ is a cyclic code of length $n$ over $R$ if and only if $\Phi(\mathscr{C})$ is an ideal of the quotient ring $R[x]/(x^n-1)$. Sometimes, we identity the cyclic code of length $n$ over $R$ with the ideal of $R[x]/(x^n-1)$.
\par
Review that $x^n-1=f_1(x)f_2(x)\cdots f_t(x)$, where, for each $l=1,2,\ldots,t$, $f_l(x)$ is a basic irreducible polynomial over $R$. Denote $\frac{x^n-1}{f_l(x)}$ by $\widehat{f}_l(x)$. Since $f_l(x)$ and $\widehat{f}_l(x)$ are coprime to each other, it follows that there are polynomials $a_l(x), b_l(x)$ in $R[x]$ such that $a_l(x)f_l(x)+b_l(x)\widehat{f}_l(x)=1$. Let $e_l(x)=b_l(x)\widehat{f}_l(x)+(x^n-1)$ and $\mathcal {R}_i=e_l(x)R[x]/(x^n-1)$. Then we have
\begin{equation}
R[x]/(x^n-1)=\mathcal {R}_1\oplus \mathcal {R}_1 \oplus \cdots\oplus \mathcal {R}_t.
\end{equation}
For any $l=1,2,\ldots,t$, the map
\begin{equation} \label{Ring isomorphism}
\begin{split}
\Psi_l:~R[x]/(f_l(x))&\rightarrow \mathcal {R}_l \\
k(x)+(f_l(x))&\mapsto (k(x)+(x^n-1))e_l(x)
 \end{split}
 \end{equation}
 is an isomorphism of rings. Therefore,
\begin{equation}  \label {eq:3.1}
R[x]/(x^n-1)\cong R[x]/(f_1(x))\times R[x]/(f_2(x))\times\cdots\times R[x]/(f_t(x)).
\end{equation}
\begin{lemma} \label{lemma:3.1}
Let $x^n-1=f_1(x)f_2(x)\cdots f_t(x)$ where, for each $l=1,2,\ldots,t$, $f_l(x)$ is a basic irreducible polynomial over $R$. Then under the map $\Psi$, the ideals of $R[x]/(f_l(x))$ are mapped into
\par
(i)~$(p^i\widehat{f}_l(x)+(x^n-1)$ for $0\leq i \leq s$;
\par
(ii)~$(p^ku\widehat{f}_l(x)+(x^n-1))$ for $0\leq k\leq s-1$;
\par
(iii)~$((p^j+\alpha u)\widehat{f}_l(x)+(x^n-1))$ for $1\leq j \leq s-1$ and $\alpha \in \mathbb{F}_p[x]/(\overline{f}_l(x))\backslash \{0\}$;
\par
(iv)~$((p^j, u)\widehat{f}_l(x)+(x^n-1))$ for $1 \leq j \leq s-1$\\
of $\mathcal {R}_l$.
\end{lemma}
\begin{proof}
Under the ring isomorphism $\Psi$, we have
\begin{equation*}
1+(f_l(x))\mapsto (1+(x^n-1))e_l.
\end{equation*}
Since $e_l=b_l(x)\widehat{f}_l(x)+(x^n-1)$, it follows that
\begin{equation*}
1+(f_l(x))\mapsto b_l(x)\widehat{f}_l(x)+(x^n-1).
\end{equation*}
Clearly, $b_l(x)\widehat{f}_l(x)+(x^n-1)\in (\widehat{f}_l(x)+(x^n-1))$.
\par
Multiplying both sides of $a_l(x)f_l(x)+b_l(x)\widehat{f}_l(x)=1$ by $\widehat{f}_l(x)$, we obtain
\begin{equation*}
b_l(x)\widehat{f}_l(x)\widehat{f}_l(x)+a_l(x)(x^n-1)=\widehat{f}_l(x).
\end{equation*}
Then
\begin{equation*}
b_l(x)\widehat{f}_l(x)\widehat{f}_l(x)+(x^n-1)=\widehat{f}_l(x)+(x^n-1),
\end{equation*}
which implies that $\widehat{f}_l(x)+(x^n-1)\in (b_l(x)\widehat{f}_l(x)+(x^n-1))$. Therefore, $(b_l(x)\widehat{f}_l(x)+(x^n-1))=(\widehat{f}_l(x)+(x^n-1))$ and the image of $(1+(f_i(x)))$ under the ring isomorphism $\Psi$ is $(\widehat{f}_l(x)+(x^n-1))$. The remainder cases can also be verified in the same way.
\end{proof}
By Theorem \ref{theorem 2.2}, Eq. (\ref{eq:3.1}) and Lemma \ref{lemma:3.1}, we have the following result directly.
\begin{theorem} \label{theorem: 3.1}
Let $x^n-1=f_1(x)f_2(x)\cdots f_t(x)$ where, for each $l=1,2,\ldots,t$, $f_l(x)$ is a basic irreducible polynomial with degree $\varepsilon_l$ over $R$. Then there are $\prod_{l=1}^t((p^{\varepsilon_l}-1)(s-1)+3s)$ cyclic codes of length $n$ over $R$. Further, any cyclic code is the sum of the ideals of $\mathcal {R}_l$.
\end{theorem}
\begin{example}
Consider a cyclic code of length $3$ over $\mathbb{Z}_4+u\mathbb{Z}_4$. Since $x^3-1=(x-1)(x^2+x+1)$, it follows that there are $(1+3\times 2)(3+3\times 2)=7\times 9=63$ cyclic codes of length $3$ over $\mathbb{Z}_4+u\mathbb{Z}_4$. Let $f_1=x-1$ and $f_2=x^2+x+1$. In the following Table 1, we list all cyclic codes of length $3$ over $\mathbb{Z}_4+u\mathbb{Z}_4$.
\begin{table}[ht]
\caption{\textbf{All cyclic codes of length $3$ over $\mathbb{Z}_4+u\mathbb{Z}_4$}}
\begin{center}
\begin{small}
\begin{tabular}{ccccc}
\hline
$0$ & $(f_2)$ & $(2f_2)$ & $(uf_2)$\\
 $(2uf_2)$ &$((2+u)f_2)$ & $(2f_2, uf_2)$ & $(f_1)$\\
$(1)$ & $(f_1, 2f_2)$ & $(f_1, uf_2)$ & $(f_1, 2uf_2)$ \\
$(f_1, (2+u)f_2)$ & $(f_1, 2f_2, uf_2)$ & $(2f_1)$ & $(2f_1, f_2)$\\
$(2)$ & $(2f_1, uf_2)$ & $(2f_1, 2uf_2)$ & $(2f_1, (2+u)f_2)$ \\
$(2, uf_2)$ & $(uf_1)$ & $(uf_1, f_2)$ & $(uf_1, 2f_2)$ \\
$(u)$ & $(uf_1, 2uf_2)$ & $(uf_1, (2+u)f_2)$ & $(u, 2f_2)$ \\
$(2uf_1)$ & $(2uf_1, f_2)$ & $(2uf_1, 2f_2)$ & $(2uf_1, uf_2)$ \\
$(2u)$ & $(2uf_1, (2+u)f_2)$ & $(2uf_1, 2f_2, uf_2)$ & $((2+u)f_1)$ \\
$((2+u)f_1, f_2)$ & $((2+u)f_1, 2f_2)$ & $((2+u)f_1, uf_2)$ & $((2+u)f_1, 2uf_2)$ \\
$(2+u)$ & $((2+u)f_1, 2f_2, uf_2)$ & $((2+xu)f_1)$ & $((2+xu)f_1, f_2)$ \\
$((2+xu)f_1, 2f_2)$ &$((2+xu)f_1, uf_2)$ & $((2+xu)f_1, 2uf_2)$ & $((2+xu)f_1, (2+u)f_2)$ \\
$((2+xu)f_1, 2f_2, uf_2)$ & $((2+(1+x)u)f_1)$ & $((2+(1+x)u)f_1, f_2)$ & $((2+(1+x)u)f_1, 2f_2)$ \\
$((2+(1+x)u)f_1, uf_2)$ & $((2+(1+x)u)f_1, 2uf_2)$ & $((2+(1+x)u)f_1, (2+u)f_2)$ & $((2+(1+x)u)f_1, 2f_2, uf_2)$ \\
$(2f_1, uf_1)$ & $(2f_1, uf_1, f_2)$ & $(2, uf_1)$ & $(2f_1, u)$\\
$(2f_1,uf_1,  2uf_2)$ &$(2f_1, uf_1, (2+u)f_2)$ & $(2, u)$ & \\

\hline
\end{tabular}
\end{small}
\end{center}
\end{table}
\end{example}
\section{Generators of cyclic codes}
In fact, the generator of the cyclic code $\mathscr{C}$ can also be characterised as follows.
\begin{theorem} \label {theorem: 3.2}
Let $\mathscr{C}$ be a cyclic code of length $n$ over $R$. Then $\mathscr{C}=(f_0(x)+uf_1(x), ug_1(x))$ with $f_0(x), f_1(x), g_1(x)\in \mathbb{Z}_q[x]$ and $g_1(x)|f_0(x)$.
\end{theorem}
\begin{proof}
Define a surjective homomorphism from $R$ to $\mathbb{Z}_q$ as $\psi(a+bu)=a$ for any $a+bu\in R$. Extend $\psi$ to the polynomial ring $R[x]/(x^n-1)$ as $\psi(a_0+a_1+\cdots+a_{n-1}x^{n-1})=\psi(a_0)+\psi(a_1)x+\cdots+\psi(a_{n-1})x^{n-1}$ for any polynomial $a_0+a_1x+\cdots+a_{n-1}\in R[x]/(x^n-1)$. Let $\mathscr{C}$ be a cyclic code of length $n$ over $R$, and restrict $\psi$ to $\mathscr{C}$. Define a set
\begin{equation*}
J=\{ f(x)\in \mathbb{Z}_q[x]/(x^n-1):~uf(x)\in {\rm Ker}\psi\}.
\end{equation*}
Clearly, $J$ is an ideal of $\mathbb{Z}_q[x]/(x^n-1)$, which implies that there is a polynomial $g_1(x)\in \mathbb{Z}_q[x]$ such that $J=(g_1(x))$. It means that ${\rm Ker}\psi=(ug_1(x))$. Further, the image of $\mathscr{C}$ under the map $\psi$ is also an ideal of $\mathbb{Z}_q[x]/(x^n-1)$. Then there is a polynomial $f_0(x)\in \mathbb{Z}_q[x]$ such that $\psi(\mathscr{C})=(f_0(x))$. Hence $\mathscr{C}=(f_0(x)+uf_1(x), ug_1(x))$. Clearly, $u(f_0(x)+uf_1(x))=uf_0(x)\in {\rm Ker}\psi$ implying $g_1(x)|f_0(x)$.
\end{proof}
\begin{lemma}\label{lemma 5}
Let $\mathscr{C}=(f_0(x)+uf_1(x), ug_1(x))$ be a cyclic code of length $n$ over $R$. If $g_1(x)$ is a monic polynomial over $R$, then we can assume that
\begin{equation*}
{\rm deg}(f_1(x))< {\rm deg}(g_1(x)).
\end{equation*}
\end{lemma}
\begin{proof}
If ${\rm deg}(g_1(x))\leq {\rm deg}(f_1(x))$, then there are polynomials $m(x), r(x)\in R[x]$ such that
\begin{equation*}
f_1(x)=g_1(x)m(x)+r(x)
\end{equation*}
with ${\rm deg}(r(x))< {\rm deg}(g_1(x))$ or $r(x)=0$. Then $\mathscr{C}=(f_0(x)+ug_1(x)m(x)+ur(x), ug_1(x))$. Let $\mathscr{C}_1=(f_0(x)+ur(x), ug_1(x))$. Clearly, $\mathscr{C}\subseteq \mathscr{C}_1$. Further, $f_0(x)+ur(x)=f_0(x)+ug_1(x)m(x)+ur(x)+u(q-1)g_1(x)m(x)$, it follows that $\mathscr{C}_1\subseteq \mathscr{C}$. Thus $\mathscr{C}=\mathscr{C}_1$.
\end{proof}
\begin{lemma}
Let $\mathscr{C}=(f_0(x)+uf_1(x), ug_1(x))$ be a cyclic code of length $n$ over $R$. If $f_0(x)=g_1(x)$, then $\mathscr{C}=(f_0(x)+uf_1(x))$. Further, if $g_1(x)$ is monic over $R$, then $(f_0(x)+uf_1(x))|(x^n-1)$.
\end{lemma}
\begin{proof}
Clearly, $(f_0(x)+uf_1(x))\subseteq \mathscr{C}$. Further, since $u(f_0(x)+uf_1(x))=uf_0(x)=ug_1(x)$, it follows that $\mathscr{C}\subseteq (f_0(x)+uf_1(x))$. Thus $\mathscr{C}=(f_0(x)+uf_1(x))$. Since $f_0(x)=g_1(x)$ and $g_1$ is monic, then $f_0(x)+uf_1(x)$ is also monic over $R$ by Lemma \ref{lemma 5}. Therefore there are polynomials $a(x), b(x)\in R[x]$ such that
\begin{equation*}
x^n-1=a(x)(f_0(x)+uf_1(x))+b(x)
\end{equation*}
with $b(x)=0$ or ${\rm deg}(b(x))< {\rm deg}(f_0(x))$. Since $b(x)\in \mathscr{C}$, it follows that $b(x)=0$. Thus $(f_0(x)+uf_1(x))|(x^n-1)$ over $R$.
\end{proof}
\begin{theorem}
Let $\mathscr{C}=(f_0(x)+uf_1(x), ug_1(x))$ be a cyclic code of length $n$ over $R$ and $f_0(x), g_1(x)$ be monic over $\mathbb{Z}_q$. Let ${\rm deg}(f_0(x))=k_0$ and ${\rm deg}(g_1(x))=k_1$. Then the set
\begin{equation*}
\beta=\{(f_0(x)+uf_1(x)), \ldots, x^{n-k_0-1}(f_0(x)+uf_1(x)), ug_1(x), \ldots, x^{k_0-k_1-1}ug_1(x)\}
\end{equation*}
forms the minimum generating set of $\mathscr{C}$, and $|\mathscr{C}|=q^{2n-k_0-k_1}$.
\end{theorem}
\begin{proof}
Let $\gamma=\{(f_0(x)+uf_1(x)), \ldots, x^{n-k_0-1}(f_0(x)+uf_1(x)), ug_1(x), \ldots, x^{n-k_1-1}ug_1(x)\}$. Then $\gamma$ spans the cyclic code $\mathscr{C}$. Further, it is sufficient to show that $\beta$ spans $\gamma$, which follows that $\beta$ also spans $\mathscr{C}$. Now we only need to show that $\beta$ is linearly independent. For simplicity, we denote $f(x)$ as $f_0(x)+uf_1(x)$. Since $f_0$ is monic, then the constant coefficient of $f(x)$ is a unit of $R$ by the fact that $f_0(x)$ is the generator polynomial of some cyclic code of length $n$ over $\mathbb{Z}_q$. Let
\begin{equation} \label{Eq.1}
a_0f(x)+a_1xf(x)+\cdots+a_{n-k_0-1}x^{n-k_0-1}f(x)=0.
\end{equation}
Let $F_0$ be the constatant coefficient of $f(x)$. Then $a_0F_0=0$, which implies that $a_0=0$. Therefore the Eq. (\ref{Eq.1}) becomes $a_1f(x)+a_2xf(x)+\cdots+a_{n-k_0-1}x^{n-k_0-2}f(x)=0$. Similarly, $a_1=a_2=\cdots=a_{n-k_0-1}=0$. Thus $(f_0(x)+uf_1(x)), x(f_0(x)+uf_1(x)), \ldots, x^{n-k_0-1}(f_0(x)+uf_1(x))$ are $R$-linear independent. One can also prove that $ug_1(x), uxg_1(x), \ldots, ux^{n-k_1-1}g_1(x)$ are $\mathbb{Z}_q$-linear independent. Thus $|\mathscr{C}|=q^{2n-k_0-k_1}$.
\end{proof}
We now consider the cyclic code $\mathscr{C}$ as a principal ideal of $R[x]/(x^n-1)$.
\begin{theorem}
Let $\mathscr{C}$ be a principal generated cyclic code of length $n$ over $R$. Then $\mathscr{C}$ is free if and only if there is a monic polynomial $g(x)$ such that $\mathscr{C}=(g(x))$ and $g(x)|(x^n-1)$. Moreover, the set
\begin{equation*}
\{g(x), xg(x), \ldots, x^{n-{\rm deg}(g(x))-1}g(x)\}
\end{equation*}
forms the minimum generating set of $\mathscr{C}$ and $|\mathscr{C}|=q^{2(n-{\rm deg}(g(x)))}$.
\end{theorem}
\begin{proof}
Suppose that $\mathscr{C}=(g(x))$ is $R$-free. Then the set $\{g(x), xg(x), \ldots, x^{n-{\rm deg}(g(x))-1}g(x)\}$ form the $R$-basis of $\mathscr{C}$. Since $x^{n-{\rm deg}g(x)}\in \mathscr{C}$, it follows that $x^{n-{\rm deg}(g(x))}$ can be written as a linear combination of the elements $g(x), xg(x), \ldots, x^{n-{\rm deg}(g(x))-1}g(x)$, i.e. $x^{n-{\rm deg}(g(x))}g(x)+\sum_{i=0}^{n-{\rm deg}(g(x))-1}a_ix^ig(x)=0$. Let $a(x)=\sum_{i=0}^{n-{\rm deg}(g(x))-1}a_ix^i+x^{n-{\rm deg}(g(x))}$. Then $a(x)g(x)=0$ in $R[x]/(x^n-1)$, which implies that $(x^n-1)|g(x)a(x)$. Since $g(x)a(x)$ is monic and ${\rm deg}(g(x)a(x))=n$, it follows that $x^n-1=g(x)a(x)$ implying $g(x)|(x^n-1)$.
\par
On the other hand, if $\mathscr{C}=(g(x))$ with $g(x)|(x^n-1)$. Then $g(x), xg(x), \ldots, x^{n-{\rm deg}(g(x))-1}g(x)$ span $\mathscr{C}$. In the following, we prove that $g(x), xg(x), \ldots, x^{n-{\rm deg}(g(x))-1}g(x)$ are $R$-linear independent. Let $g(x)=g_0+g_1x+\cdots+x^{{\rm deg}(g(x))}$ and
\begin{equation}\label{Eq.2}
a_0g(x)+a_1xg(x)+\cdots+a_{n-{\rm deg}(g(x))-1}x^{n-{\rm deg}(g(x))-1}g(x)=0.
\end{equation}
Since $g_0$ is a unit of $R$ and $a_0g_0=0$, it follows that $a_0=0$. Then the Eq. (\ref{Eq.2}) becomes
\begin{equation*}
a_1g(x)+a_2xg(x)+\cdots+a_{n-{\rm deg}(g(x))-1}x^{n-{\rm deg}(g(x))-2}g(x)=0.
\end{equation*}
Similarly, we have $a_1=a_2=\cdots=a_{n-{\rm deg}(g(x))-1}=0$. Thus $g(x), xg(x), \ldots, x^{n-{\rm deg}(g(x))-1}g(x)$ are $R$-linear independent, i.e. $\mathscr{C}$ is free over $R$.
\end{proof}
Similar to the cyclic codes over finite fields, we can also give the following BCH-type bound on the minimum Hamming distance of free cyclic codes over $R$. The proof precess is similar to that of the case over finite fields. Here, we omit it.
\begin{theorem}{(BCH-type Bound)}
Let $\mathscr{C}=(g(x))$ be a free cyclic code of length $n$ over $R$. Suppose that $g(x)$ has roots $\xi^b, \xi^{b+1}, \ldots, \xi^{b+\delta-2}$, where $\xi$ is a basic primitive $n$th root of unity in some Galois extension ring of $R$. Then the minimum Hamming distance of $\mathscr{C}$ is $d_H(\mathscr{C})\geq \delta$.
\end{theorem}
\begin{example}
Let $R=\mathbb{Z}_8+u\mathbb{Z}_8$. Suppose that $f(x)=x^4+4x^3+6x^2+3x+1$, then $f(x)$ is a basic primitive polynomial with degree $4$ over $R$. Let $\xi=x+(f(x))$. Then $\xi$ is a basic primitive element of $\mathcal {R}=R[x]/(f(x))$. Consider a cyclic code $\mathscr{C}$ of length $15$ and generated by the polynomial $g(x)=x^{10}+6x^9+x^8+6x^7+3x^5+7x^4+4x^3+7x^2+5x+1$. Then $g(x)|(x^{15}-1)$, which implies that $\mathscr{C}$ is a free cyclic code and $|\mathscr{C}|=64^5$. Furthermore, $g(x)$ has $\xi, \xi^2, \xi^3, \xi^4, \xi^5, \xi^6$ as its part roots in $\mathcal {R}[x]$. Then, by Theorem 7, we have that $d_H(\mathscr{C})\geq 7$. Since $d_H(4g(x))=7$, it follows that $d_H(\mathscr{C})=7$, i.e. $\mathscr{C}$ is a $(15, 64^5, 7)$ cyclic code over $R$.
\end{example}

\begin{example}
Let $\mathscr{C}=(1+2x+x^2+3x^3, ux-u)$ be a cyclic code of length $7$ over $\mathbb{Z}_4+u\mathbb{Z}_4$. Then, by Theorem 6, we have that $|\mathscr{C}|=4^{10}$. For any $a+bu\in \mathbb{Z}_4+u\mathbb{Z}_4$, define $\varphi(a+bu)=(b, a+b)$ and Lee weight of $a+bu$ to be the Lee weight of $(b, a+b)$. Then $\varphi$ is a preserving Lee distance $\mathbb{Z}_4$-linear map from $\mathbb{Z}_4+u\mathbb{Z}_4$ to $\mathbb{Z}_4^2$ (see \cite{Yildiz}). By the help of Magma computational software, we have that $\varphi(\mathscr{C})$ is a $(14, 4^{10}, 4)$ linear code over $\mathbb{Z}_4$, which binary image under the Gray map from $\mathbb{Z}_4$ to $\mathbb{F}_2^2$ is a $(28, 2^{20}, 4)$ binary code having the same parameters as an optimal binary linear code $[28, 20, 4]$.
\end{example}
\section*{Acknowledgments}
This research is supported by the National Key Basic Research Program of China (Grant No. 2013CB834204), and the National Natural Science Foundation of China (Grant Nos. 61171082 and  61301137).

\end{document}